\documentclass[a4paper, 11pt]{article}

\usepackage[utf8]{inputenc}

\usepackage{xcolor}
\usepackage[linesnumbered,shortend]{algorithm2e}
\usepackage{lipsum}
\usepackage{listings}
\usepackage{comment}

\usepackage{amsmath}
\usepackage{amssymb}
\usepackage{amsthm}
\usepackage{esvect}
\usepackage{url}
\newtheorem{theorem}{Theorem}
\newtheorem{lemma}[theorem]{Lemma}

\usepackage{thmtools}

\usepackage{paralist}    

\usepackage{xspace}
\usepackage{graphicx}

\usepackage[normalem]{ulem}

\usepackage{tikz}
\usetikzlibrary{tikzmark, calc, positioning, shapes.misc}
\definecolor{greencustom}{rgb}{0,0.65,0.2}

\lstdefinestyle{myC}{
  basicstyle=\footnotesize,
  language=C++,
  tabsize=2,
  keywordstyle=\color{blue}\ttfamily,
  stringstyle=\color{red}\ttfamily,
  commentstyle=\color{gray}\ttfamily, 
  morecomment=[l][\color{magenta}]{\#},
  mathescape=true,
  numberstyle=\tiny\ttfamily,                 
  numbers=left,                    
  numbersep=5pt,
}
\def\inc{+\hspace{-0.4em}=}

\newcommand{\bigo}{\mathcal{O}}
\newcommand{\kmin}{\ensuremath{k_{\min}}\xspace}
\newcommand{\kmax}{\ensuremath{k_{\max}}\xspace}

\usepackage{todonotes}
\setuptodonotes{inline}

\title{Tightening I/O Lower Bounds\\through the Hourglass Dependency Pattern}

\author{
    Lionel Eyraud-Dubois\footnote{Inria, Univ. Bordeaux, CNRS, Bordeaux INP, LaBRI, UMR 5800, lionel.eyraud-dubois@inria.fr}
	\and 
	Guillaume Iooss\footnote{Univ. Grenoble Alpes, Inria, CNRS, Grenoble INP, LIG, guillaume.iooss@inria.fr}
	\and
	Julien Langou\footnote{University of Colorado Denver, USA, Centre Inria de Lyon, France, julien.langou@ucdenver.edu}
	\and
	Fabrice Rastello\footnote{Univ. Grenoble Alpes, Inria, CNRS, Grenoble INP, LIG, fabrice.rastello@inria.fr}
}

\date{April 2024}

\begin{document}



\maketitle


\begin{abstract}
When designing an algorithm, one cares about arithmetic/computational complexity, but data movement (I/O) complexity plays an increasingly important role that highly impacts performance and energy consumption.
For a given algorithm and a given I/O model, scheduling strategies such as loop tiling can reduce the required I/O down to a limit, called the I/O complexity, inherent to the algorithm itself. 

The objective of I/O complexity analysis is to compute, for a given program, its minimal I/O requirement among all valid schedules. 
We consider a sequential execution model with two memories, an infinite one, and a small one of size $S$ on which the computations retrieve and produce data.
The I/O is the number of reads and writes between the two memories.

We identify a common ``{\em hourglass pattern}'' in the dependency graphs of several common linear algebra kernels.
Using the properties of this pattern, we mathematically prove tighter lower bounds on their I/O complexity, which improves the previous state-of-the-art bound by a parametric ratio.
This proof was integrated inside the IOLB automatic lower bound derivation tool.
\end{abstract}


\section{Introduction}

When designing an algorithm, we usually reason about its computational complexity, to estimate the increase in its execution time, when its problem sizes increase.
However, the amount of computation is not the only factor when estimating the performance of an algorithm.
The data movement (I/O) also plays an increasingly important role in both performance and the energy consumed by an algorithm.

In order to minimize the amount of data movement of an algorithm, we modify its schedule by using program transformations, such as the tiling transformation~\cite{Irigoin88_Tiling}.
However, finding an optimal schedule is not trivial, due to the size and complexity of the optimization space.
The notion of \emph{I/O complexity} of an algorithm indicates the minimal amount of data movement required for any valid schedule of an algorithm.
It provides an indication on how far it is possible to optimize the I/O of an algorithm.

However,  due to the number of possible valid schedules, it is not possible to directly compute the I/O complexity.
Instead, we search for a lower and upper bound on the minimal amount of data movement.
To find an upper bound, one only needs to exhibit a valid schedule and compute its I/O cost.
However, finding a lower bound requires us to reason over all possible schedules, which requires a mathematical proof based on the data dependency graph of the program.

Several proof techniques exist for deriving lower bounds, such as the wavefront~\cite{Elango14} or the $K$-partitioning technique~\cite{hong-pebble}.
Depending on the shape of the dependencies, some of these techniques might work better than the others: for example, the wavefront technique is usually the most effective for stencil-like computation, while the $K$-partitioning technique is usually better for linear-algebraic computation.
Other algorithms might require a specialized proof to obtain an asymptotically tight bound (e.g., for the SYRK kernel~\cite{Beaumont22_IOcompl_syrk}).
However, there are still some kernels for which the ratio between their proven lower bound and their best upper bound is parametric. For example, such a ratio is described in Table 1 of ~\cite{Olivry_pldi20} for the kernels of the Polybench benchmark suite~\cite{polybench}.

Some of the proof techniques mentioned above have been automatized and implemented in automatic data movement lower bound derivation tools, such as IOLB~\cite{Olivry_pldi20}, so that they can be applied to any kernel given as an input.

\paragraph{Contributions} In this paper, we consider several important linear algebra kernels whose dependence graph exhibits a dependency pattern called the \emph{hourglass pattern}.
We propose a new proof technique that uses the properties of this pattern to improve the lower bound on the minimal data movement required by these kernels.

In more detail:
\begin{itemize}
    \item We define the \emph{hourglass pattern}, a pattern of the dependencies of a program, and present its properties.

    \item We provide a lower bound derivation proof, based on an adaptation of the $K$-partitioning technique, that tighten the lower bound of a program exhibiting an hourglass pattern.
    This proof has been fully automatized, inside the tool IOLB~\cite{Olivry_pldi20, Olivry_pldi21}.
    
    \item We present new data movement lower bounds for several linear algebra kernels: Modified GramSchmidt (MGS), QR Householder (GEQR2 and ORG2R in the LAPACK~\cite{githubLAPACK} library; also called A2V and V2Q), bidiagonal matrix reduction (GEBD2) and Hessenberg matrix reduction (GEHD2).
    For all of these kernels, their asymptotic bound was improved by a parametric factor, compared to the bound obtainable by the classical $K$-partitioning technique.
    Figure~\ref{fig:recap_bounds} and Figure~\ref{fig:IOLB_full_bounds} summarize the new lower bounds found for several linear algebra kernels.

    \item We also provide tiled orderings for MGS and Householder, resulting in upper bounds that asymptotically match these new lower bounds.
    This proves the optimality of the new I/O lower bounds.
\end{itemize}

\paragraph{Outline}
In Section~\ref{sec:background}, we provide some background on the I/O complexity and introduce the $K$-partitioning method for deriving a lower bound on the minimal amount of data movement of a computation.
In Section~\ref{sec:hourglass_pattern}, we present the \emph{hourglass pattern}, a pattern of dependencies whose properties can be used to improve the derived lower bound.
In Section~\ref{sec:gen_iolb_proof}, we show how to exploit the hourglass pattern to adapt the $K$-partitioning method, in order to obtain a tighter bound.
In Section~\ref{sec:new_bounds}, we list different linear algebra kernels exhibiting an hourglass pattern, and their associated improved lower bound.
Additionally, Annex~\ref{sec:annex_ub_algo} contains the tiled algorithm for two of these kernels. The data movement of these algorithms provides an upper bound to the minimal amount of data movement, which matches asymptotically their new lower bound.

\section{Background - I/O complexity and the K-partitioning method}
\label{sec:background}

In this section, we present a state-of-the-art method of proof -- the $K$-partitioning method -- which infers a lower bound on the amount of data movement needed by a computation.

\paragraph{Memory model and I/O complexity}
We consider a simple two-level memory model, composed of (1) a slow memory of unbounded size, and (2) a fast memory of size $S$.
Both memories can transmit data from one to another, as long as the constraint on the size of the small memory is satisfied.
When we perform an operation, the data used must be present in the small memory, and the data produced must be committed in the small memory.

We consider a program, which performs a collection of operations organized in \emph{statements}.
Each statement $SX$ has multiple \emph{instances} $SX[\vec{i}]$, where $\vec{i}$ is a vector of the surrounding loop indexes.

In this paper, we consider a subclass of programs called \emph{polyhedron} (or affine) \emph{programs}.
These programs are combinations of nested loops and statements such as:
(i) the loop bounds are affine constraints using the surrounding loop indexes and the program parameters (e.g., the sizes of an input array);
(ii) the array accesses are affine expression of the surrounding loop indexes and program parameters. All the programs presented in this paper, such as Figure~\ref{fig:householder_a2v}, are polyhedral programs.
Also, we will call a quantity ``parametric'' when it is a function of the parameters of the program, which are considered as symbolic constants.

A statement instance might depend on the data produced by another instance, for example when the first instance uses a value produced by the second instance.
We call this a \emph{dependency} between these two statement instances.
These dependencies impose constraints on the order of execution of the program.
This order of execution is called the \emph{schedule} or the \emph{ordering}.
Frequently, the data consumed/produced by the statements of a program are too big to fit all at once in the small memory, it it thus necessary to \emph{spill}, i.e., to transfer some data back into the slow memory and retrieve it later when needed.
However, doing so increases the amount of data movement between both memories.

Given a valid schedule for a program, the \emph{I/O cost} for this program and for this ordering is the amount of data movement required, i.e. the number of data transfers between both memories.
The \emph{I/O complexity} of a program is the minimal I/O cost that can be reached by any valid schedule.
%
This quantity is interesting, in particular in the context of an architecture where the transfer of data is the limiting factor for performance.
Knowing the minimal amount of data transfers is a good algorithmic indicator to know if it could be theoretically optimized further.
%
However, because we need to find the minimal I/O cost \emph{for all possible orderings}, the exact I/O complexity is hard to evaluate.
Instead, we rely on bounds on the I/O complexity of a program: we can provide a mathematical proof for the lower bound, and exhibit an ordering (i.e., an implementation of a program) that reaches an I/O cost and provides an upper bound.

\paragraph{CDAG and red-white pebble game}
When trying to prove a lower bound, one should decide whether redundant computation is allowed or not. 
The \emph{red-white pebble game}, a variation of the red-blue pebble game of Hong and Kung~\cite{hong-pebble}, was introduced by Olivry et al.~\cite{Olivry_pldi20} in order to model the state of the memories during the execution of a program \emph{without} recomputation, which matched the assumptions we make in this paper.

This game is played on the \emph{Computational Directed Acyclic Graph} (CDAG) of the program. This is a directed graph $G$, where
\begin{itemize}
    \item the nodes $V$ represent the computation (statement instance) of a program, and
    \item the edges represent the flow dependencies between the computations of the program.
\end{itemize}
Notice that the inputs of a program are nodes that do not have incoming edges. The outputs of a program are a subset of nodes $O \subset V$; they might have outgoing edges.

During a \emph{red-white pebble game}, red and white pebbles are placed on the nodes of a CDAG.
A \emph{white pebble} represents a computation that was performed, and a \emph{red pebble} represents a computation whose output is currently stored in the small memory.
A game follows this set of rules:
\begin{itemize}
    \item At the start, the only pebbles in the CDAG are white pebbles, placed on the inputs of the program.

    \item At most $S$ red pebbles can be simultaneously present on the nodes of the CDAG.

    \item \textbf{Spill:} a red pebble can be removed from a node.

    \item \textbf{Compute:} When a node does not have a white pebble, but all its predecessors have red pebbles, then we can place both a white and a red pebble on it.

    \item \textbf{Load:} A red pebble can be added to nodes with a white pebble.

    \item The game ends when each node has a white pebble on it.
\end{itemize}
Notice that once a white pebble is placed on a node, it cannot be removed. This prevents recomputation.
In order to compute the amount of data movements, we focus on the number of red pebbles added with the \textbf{Load} rule during a game.
This means that we only focus on the ``Load'' portion of the data movements and ignore its ``Store'' part.
The resulting bounds are still valid, and because the number of ``Load'' often dominates the number of ``Store'', their tightness should not be strongly impacted.
This assumption is identical to the one made in~\cite{Olivry_pldi20}.

\paragraph{$K$-partitioning method} The \emph{$K$-partitioning} method introduced in the seminal paper of Hong and Kunk~\cite{hong-pebble} is a proof technique that allows to derive a lower bound.

The first idea is to consider a partition of the CDAG and games that play on each set of the partition one by one.

Then, we consider the notion of \emph{$K$-bounded set}.
An \emph{inset} of a set $E$ of nodes of the CDAG, noted $InSet(E)$, is the set of data used by $E$ but not produced by a computation of $E$.
A $K$-bounded set is a set of nodes $E$ of the CDAG whose inset has a size at most $K$: $|InSet(E)| \leq K$.
This notion is interesting because an $(S+T)$-bounded set requires at least $T$ additional data to fit in the small memory (of size $S$). Thus, even if the small memory is filled with interesting data, we will need at least $T$ load operations to perform the computations of this set.
%
In addition, we assume that our $K$-bounded sets are \emph{convex}: if there is a dependency chain between two points of a $K$-bounded set $E$, then all the intermediate points must belong to $E$.

Finally, a \emph{$K$-partition} is a partition into convex $K$-bounded sets.
The idea is to consider all $K$-partitions of a CDAG and to count how many sets are in this partition.
By choosing $K=S+T$, we know that there will be at least $T$ loads per set of the partition.
Therefore, a lower bound on the number of loads is $T$ times the minimal number of sets in such a partition.

\begin{theorem}[$(S+T)$-partitioning I/O lower bound~\cite{Elango15}]
Let $S$ be the size of the small memory, and for any $T>0$ let $U$ be the maximal size of a $(S+T)$-partition. Let $V$ be the set of nodes of the CDAG of the program.
%
Then, a lower bound on the number $Q$ of data movement of the program is:
$$
 T.\left\lfloor\frac{|V|}{U}\right\rfloor\leq Q
$$
\label{thm:STpartitioning}
\end{theorem}

Then, we pick a value of $T$ (which means a value of $K$) that leads to the tightest lower bound.

To estimate the minimal number of sets in a $K$-partition, we can estimate the maximum size of a set inside this partition.
In other words, an upper bound on the size of a $K$-bounded set can be transformed into a lower bound on the amount of data movement required.

\paragraph{Upper bound on the size of a $K$-bounded set}
An upper bound on the size of a $K$-bounded set $E$ can be obtained by analyzing the dependencies of the program.
Indeed, for a polyhedral program, dependencies between its statement instances are associated with affine relations, matching the loop indices of the data producing instance with the data consuming instance.
When examining the path of affine dependencies starting from any node of $E$ to a node of the inset of $E$, we can either obtain a projection or a translation. In both cases, the image of $E$ through these affine functions $\phi$ can be mapped to geometrical borders or projections of $E$, and can be associated with parts of $InSet(E)$.
This is the key geometrical intuition that leads us to use the Brascamp-Lieb theorem.

The Brascamp-Lieb theorem is a geometrical way to bound the volume of a set by the volume of its projections, which can be bounded by $K$.
\begin{theorem}[Brascamp-Lieb theorem~\cite{Christ13}]
Let $d$ and $d_j$ be non-negative integers and $\phi_j: \mathbb{Z}^d \mapsto \mathbb{Z}^{d_j}$ be a collection of group homomorphisms for all $1 \leq j \leq m$.

If we have a collection of coefficients $s_j \in [0,1]$ such that, for any subgroup $\mathcal{H} \subset \mathbb{Z}^d$:
$$rank(\mathcal{H}) \leq \sum_{j=1}^{m} s_j \times rank(\phi_j(\mathcal{H})).$$

Then, for any non-empty finite set $E \subset \mathbb{Z}^d$:
$$|E| \leq \prod_{j=1}^{m} |\phi_j(E)|^{s_j}.$$
\label{thm:BLthm}
\end{theorem}


For example, if we consider a 3D set and the 3 canonical projections $\phi_{j,k}(i,j,k)=(j,k)$, $\phi_{i,k}(i,j,k)=(i,k)$ and $\phi_{i,j}(i,j,k)=(i,j)$, this theorem gives us the following inequality between the volume of $E$ and the area of its faces $\phi_x(E)$:

$$
|E| \leq |\phi_{j,k}(E)|^{1/2} \times |\phi_{i,k}(E)|^{1/2} \times |\phi_{i,j}(E)|^{1/2}.
$$

As another example, we could consider instead the projections $\phi_{i}(i,j,k) = (i)$, $\phi_{j}(i,j,k) = (j)$ and $\phi_{k}(i,j,k) = (k)$, to obtain the following inequality:
$$
|E| \leq |\phi_{i}(E)| \times |\phi_{j}(E)| \times |\phi_{k}(E)|.
$$

In the case of a $K$-bounded set, we consider the path of dependencies to automatically derive these projections $\phi \in \Phi$.
So, $\phi(E)$ can be mapped to one part of the inset of $E$, and its size is bounded by $K$.

\section{The hourglass pattern}
\label{sec:hourglass_pattern}

In this section, we describe the intuition of our core contribution.
We consider a specific pattern of dependencies, called the \emph{hourglass pattern}, that forces a convex $K$-bounded set to have a specific shape.
We can exploit this property to significantly improve the derived I/O complexity lower bound of programs that exhibit such a pattern.

In the whole section, we use the Modified Gram-Schmidt algorithm as an illustrative example, whose right-looking variant is provided in Figure~\ref{fig:mgs_rl}.
Using the automatic tool IOLB~\cite{Olivry_pldi20} to apply the $K$-partitioning method (described in Section~\ref{sec:background}) to the MGS computation results in a lower bound in $\Omega\left(\frac{MN^2}{\sqrt{S}}\right)$.
By using the hourglass pattern, we obtain a more precise lower bound:
$$
\frac{M^2N(N-1)}{8(S+M)} \leq Q(MGS) 
$$


\begin{figure}
\lstset{
  basicstyle=\footnotesize,
  xleftmargin=.07\textwidth, xrightmargin=.05\textwidth
}
\begin{lstlisting}[style=myC]
for ($k=0$; $k<N$; $k\ \inc 1$) {
  $nrm$ = $0.0$;
  for ($i=0$; $i<M$; $i\ \inc 1$)
    $nrm$ += $A[i][k]$ * $A[i][k]$;
  $R[k][k]$ = sqrt$(nrm)$;
  
  for ($i = 0$; $i < M$; $i\ \inc 1$)
    $Q[i][k]$ = $A[i][k]$ / $R[k][k]$;
  
  for ($j = k + 1$; $j < N$; $j\ \inc 1$) {
    $R[k][j]$ = $0.0$;
    for ($i = 0$; $i < M$; $i\ \inc 1$)
SR:    $R[k][j]$ += $Q[i][k]$ * $A[i][j]$;
    for ($i = 0$; $i < M$; $i\ \inc 1$)
SU:    $A[i][j]$ = $A[i][j]$ - $Q[i][k]$ * $R[k][j]$;
  }
}
\end{lstlisting}
\caption{Modified Gram-Schmidt - Right-Looking (from Polybench~\cite{polybench}). The input matrix $A$ is of size $M \times N$, and the output of the algorithm are matrices $Q$ (the orthonormalized column vector basis) and $R$ such that $A=QR$. The usual right-looking Gram-Schmidt reuses the matrix $A$, instead of defining a new matrix $Q$.
$SR$ and $SU$ are labels of two statements, updating $R$ and $A$.}
\label{fig:mgs_rl}
\end{figure}

\subsection{Intuition of the hourglass pattern}
\label{subsec:hourglass_intuition}

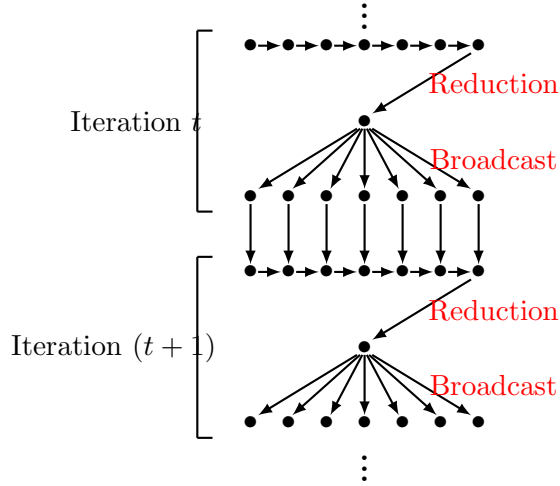
\begin{figure}
\begin{center}
\begin{tikzpicture}
    \node at (1.5,0.5) {\textbf{\vdots}};

    \node at (0,0) {$\bullet$};
    \node at (0.5,0) {$\bullet$};
    \node at (1,0) {$\bullet$};
    \node at (1.5,0) {$\bullet$};
    \node at (2,0) {$\bullet$};
    \node at (2.5,0) {$\bullet$};
    \node at (3,0) {$\bullet$};

    \draw[-latex,thick] (0.1,0) -- (0.4,0);
    \draw[-latex,thick] (0.6,0) -- (0.9,0);
    \draw[-latex,thick] (1.1,0) -- (1.4,0);
    \draw[-latex,thick] (1.6,0) -- (1.9,0);
    \draw[-latex,thick] (2.1,0) -- (2.4,0);
    \draw[-latex,thick] (2.6,0) -- (2.9,0);
    \draw[-latex,thick] (2.9,-0.1) -- (1.6,-0.9);

    \node[red] at (3.2,-0.5) {Reduction};
    
    \node at (1.5,-1) {$\bullet$};

    \draw[latex-,thick] (0.1,-1.9) -- (1.4,-1.1);
    \draw[latex-,thick] (0.55,-1.9) -- (1.44,-1.1);
    \draw[latex-,thick] (1.05,-1.9) -- (1.47,-1.1);
    \draw[latex-,thick] (1.5,-1.9) -- (1.5,-1.1);
    \draw[latex-,thick] (1.95,-1.9) -- (1.53,-1.1);
    \draw[latex-,thick] (2.45,-1.9) -- (1.56,-1.1);
    \draw[latex-,thick] (2.9,-1.9) -- (1.6,-1.1);

    \node[red] at (3.2,-1.5) {Broadcast};
    
    \node at (0,-2) {$\bullet$};
    \node at (0.5,-2) {$\bullet$};
    \node at (1,-2) {$\bullet$};
    \node at (1.5,-2) {$\bullet$};
    \node at (2,-2) {$\bullet$};
    \node at (2.5,-2) {$\bullet$};
    \node at (3,-2) {$\bullet$};

    \node at (-1.5,-1) {Iteration $t$};
    \draw[thick] (-0.7,0.2) -- (-0.7,-2.2);
    \draw[thick] (-0.7,0.2) -- (-0.5,0.2);
    \draw[thick] (-0.7,-2.2) -- (-0.5,-2.2);

    \draw[-latex,thick] (0,-2.1) -- (0,-2.9);
    \draw[-latex,thick] (0.5,-2.1) -- (0.5,-2.9);
    \draw[-latex,thick] (1,-2.1) -- (1,-2.9);
    \draw[-latex,thick] (1.5,-2.1) -- (1.5,-2.9);
    \draw[-latex,thick] (2,-2.1) -- (2,-2.9);
    \draw[-latex,thick] (2.5,-2.1) -- (2.5,-2.9);
    \draw[-latex,thick] (3,-2.1) -- (3,-2.9);

    \node at (0,-3) {$\bullet$};
    \node at (0.5,-3) {$\bullet$};
    \node at (1,-3) {$\bullet$};
    \node at (1.5,-3) {$\bullet$};
    \node at (2,-3) {$\bullet$};
    \node at (2.5,-3) {$\bullet$};
    \node at (3,-3) {$\bullet$};

    \draw[-latex,thick] (0.1,-3) -- (0.4,-3);
    \draw[-latex,thick] (0.6,-3) -- (0.9,-3);
    \draw[-latex,thick] (1.1,-3) -- (1.4,-3);
    \draw[-latex,thick] (1.6,-3) -- (1.9,-3);
    \draw[-latex,thick] (2.1,-3) -- (2.4,-3);
    \draw[-latex,thick] (2.6,-3) -- (2.9,-3);
    \draw[-latex,thick] (2.9,-3.1) -- (1.6,-3.9);

    \node[red] at (3.2,-3.5) {Reduction};
    
    \node at (1.5,-4) {$\bullet$};

    \draw[latex-,thick] (0.1,-4.9) -- (1.4,-4.1);
    \draw[latex-,thick] (0.55,-4.9) -- (1.44,-4.1);
    \draw[latex-,thick] (1.05,-4.9) -- (1.47,-4.1);
    \draw[latex-,thick] (1.5,-4.9) -- (1.5,-4.1);
    \draw[latex-,thick] (1.95,-4.9) -- (1.53,-4.1);
    \draw[latex-,thick] (2.45,-4.9) -- (1.56,-4.1);
    \draw[latex-,thick] (2.9,-4.9) -- (1.6,-4.1);

    \node[red] at (3.2,-4.5) {Broadcast};
    
    \node at (0,-5) {$\bullet$};
    \node at (0.5,-5) {$\bullet$};
    \node at (1,-5) {$\bullet$};
    \node at (1.5,-5) {$\bullet$};
    \node at (2,-5) {$\bullet$};
    \node at (2.5,-5) {$\bullet$};
    \node at (3,-5) {$\bullet$};

    \node at (-1.8,-4) {Iteration $(t+1)$};
    \draw[thick] (-0.7,0.2-3) -- (-0.7,-5.2);
    \draw[thick] (-0.7,0.2-3) -- (-0.5,0.2-3);
    \draw[thick] (-0.7,-5.2) -- (-0.5,-5.2);

    \node at (1.5,-5.5) {\textbf{\vdots}};
\end{tikzpicture}
\end{center}
\caption{Shape of an hourglass pattern, inside the dependence graph. A node is an instance of a statement of the program, and an edge is a data dependency between two nodes. The $t$ dimension is an external loop surrounding the hourglass.}
\label{fig:hourglass_pattern}
\end{figure}

\paragraph{Intuition} Figure~\ref{fig:hourglass_pattern} presents the main idea of the hourglass pattern.
It is a repeating succession of reduction and broadcast statements, such that the number of elements reduced/broadcasted is parametric, thus greater than the cache size $S$.
There are 3 categories of dimensions in this pattern:
(a) the dimensions over which the reduction and the broadcast are performed (horizontal axis of Figure~\ref{fig:hourglass_pattern}),
(b) the ``temporal'' dimensions over which the hourglass pattern is repeated (vertical axis of Figure~\ref{fig:hourglass_pattern}), and
(c) the neutral dimensions that do not interact with the hourglass pattern.

\paragraph{Running example} The hourglass pattern appears on several linear algebra kernels, including the Modified Gram-Schmidt kernel (Figure~\ref{fig:mgs_rl}).
The pattern appears between the last two statements: statement $SR$ which updates $R[k][j]$, and statement $SU$ which updates $A[i][j]$.
The statement $SR$ is a reduction along the $i$ dimension and uses, in particular, all the values of $A[\cdot][j]$ produced during the previous iteration of $k$.
The statement $SU$ broadcasts $R[k][j]$ across the $i$ dimension to update all the $A[\cdot][j]$ of the current iteration of $k$.
Therefore, dimension $k$ is a temporal dimension, dimension $i$ is the reduction/broadcast dimension and dimension $j$ is a neutral dimension. There is exactly one dimension in each category in this example, but in general, there might be several.

\paragraph{Consequences of the hourglass pattern}
When considering a $K$-bounded set over this pattern, we notice that if it spans over several iterations of the temporal dimension $t$, then the set \textbf{must} include all the nodes of the broadcast/reduction in between, due to the convexity property of the set.
Therefore, we have two situations:
\begin{itemize}
    \item Either the $K$-bounded set spans over several iterations of $t$ and includes all the nodes over the reduction/broadcast dimension. Notice that this is not always possible, depending on the size of the broadcast/reduction dimension and the value of $K$.
    \item Or, the $K$-bounded set is ``flat'' along the $t$ dimension.
\end{itemize}
For both situations, we can deduce much stronger constraints on the sizes of the projection $|\phi(E)|$ used in the Brascamp-Lieb theorem (Theorem~\ref{thm:BLthm}).
This provides the derivation of an improved lower bound compared to the classical methodology.

\subsection{Hourglass pattern - formal definition}
\label{subsec:hourglass_def}

\begin{figure}
\lstset{
  basicstyle=\footnotesize,
  xleftmargin=.07\textwidth, xrightmargin=.05\textwidth
}
\begin{lstlisting}[style=myC]
for($k = 0$; $k < N$; $k\ \inc 1$){
   $norma2$ = $0.0$;
   for($i = k+1$; $i < M$; $i\ \inc 1$){
      $norma2$ += $A[i][k]$ * $A[i][k]$;
   }
   $norma$ = $sqrt(A[k][k]$ * $A[k][k]$ + $norma2)$;
   $A[k][k]$=($A[k][k] > 0$)?
         ($A[k][k]$+$norma$):($A[k][k]$-$norma$);
   
   $tau[k]$ = $2.0$ /
       ($1.0$ + $norma2$ / ($A[k][k]$ * $A[k][k]$));
   
   for($i = k+1$; $i < M$; $i\ \inc 1$){
      $A[i][k]$ /= $A[k][k]$;
   }
   $A[k][k]$ = ($A[k][k] > 0$)?($-norma$):($norma$);
   
   for($j = k+1$; $j < N$; $j\ \inc 1$){
      $tau[j]$ = $A[k][j]$;
      for($i = k+1$; $i < M$; $i\ \inc 1$){
 SR:      $tau[j]$ += $A[i][k]$ * $A[i][j]$;
      }
      $tau[j]$ = $tau[k]$ * $tau[j]$;
      $A[k][j]$ = $A[k][j]$ - $tau[j]$;
      for($i = k+1$; $i < M$; $i\ \inc 1$){
 SU:      $A[i][j]$ = $A[i][j]$ - $A[i][k]$ * $tau[j]$;
      }
   }
}
\end{lstlisting}
\caption{QR Householder computation - Part A2V (LAPACK routine GEQR2).}
\label{fig:householder_a2v}
\end{figure}

In this section, we provide a formal definition of the hourglass pattern.

\paragraph{Preliminary notations}
Given an instance for a statement $SX$, we call an \emph{iteration vector} the tuple of the (integral) values of its surrounding loop indices.
The \emph{iteration domain $\mathcal{D}_{SX}$ of the statement $SX$} is the set of iteration vectors that respect the conditions on the indices of the surrounding loops.

As mentioned above, in general, there might be several reduction or temporal dimensions. Thus, we will consider sets of dimensions, and represent their iteration as a vector $\vec{i} = (i_1, i_2, \dots)$.
Given an iteration $\vv{k}$, we write  $\vv{k+1}$ to represent the next valid lexicographic value of $\vv{k}$. We extend this notation to $\vv{k+n}$ where $n$ is an integer.

\paragraph{The hourglass pattern}
Considering a statement $S$ of the CDAG of a program, the \emph{hourglass pattern} is a pattern of dependencies with the following properties:
\begin{itemize}
    \item \emph{Partitioning of the dimensions.} The dimensions of the statement $SX$ can be partitioned into 3 groups: (i) the temporal dimensions $\vec{k}$, (ii) the reduction/broadcast dimensions $\vec{i}$, and (iii) the neutral dimensions $\vec{j}$.
    For simplicity of the presentation, we assume that $\vec{k}$ are the first/outer dimensions and that $\vec{i}$ are the last/inner dimensions.
    
    \item \emph{Path in the dependence graph.} For any valid value of $\vec{i}$ and $\vec{i'}$, there is a dependency chain between the instances $SX[\vec{k},\vec{j},\vec{i}]$ and $SX[\vv{k+1}, \vec{j},$ $\vec{i'}]$.
    
    \item \emph{Large width of the hourglass.} Let us consider $W$, the number of statement instances on all the dependency chains between $SX[\vec{k},\vec{j}, \vec{i}]$ and $SX[\vv{k+2},\vec{j}, \vec{i}]$. This expression depends on the parameters of the program, and cannot be bounded by a constant value.
\end{itemize}

Notice that, to have such a chain of dependencies, it must include a reduction and a broadcast.
The dimensions $\vec{i}$ are the dimensions which are reduced on and broadcasted over.
The dimensions $\vec{k}$ are the dimensions that are incremented by a constant factor when looping along this loop.
So, once a path is found, partitioning the dimensions should be unambiguous, if this condition is also satisfied.

The automatic detection of such an hourglass pattern has been implemented inside the IOLB tool, using a polyhedral library~\cite{isl,sven-barvinok}.

\paragraph{Examples}
For the MGS computation (Figure~\ref{fig:mgs_rl}), we consider the statement $SU$ and the cycle of dependencies going through the $SR$.
We confirm that there are $2M$ statement instances inside a dependency chain between two instances $SU[k,j,i]$ and $SU[k+2,j,i]$: $SR[k+1,j,\cdot]$ and $SU[k+1,j,\cdot]$.
The same reasoning would hold if we considered the statement $SR$ instead of $SU$.

For the A2V QR Householder computation (Figure~\ref{fig:householder_a2v}), we consider the statement $SU$, and the cycle of dependencies going through the $SR$ statement.
There are $(M-k)$ statement instances $SR[k,j,\cdot]$, inside a dependency chain between two instances $SU[k,j,i]$ and $SU[k+1,j,i]$.


\section{Lower bound proof using the hourglass pattern}
\label{sec:gen_iolb_proof}

In this section, we show how to exploit the hourglass pattern of a program to derive a tighter lower bound on the data movement.
We will use the Modified Gram-Schmidt program (Figure~\ref{fig:mgs_rl}) as a running example to illustrate our proof.

Once again, this proof has been integrated inside the automatic data movement lower bound derivation tool, IOLB~\cite{Olivry_pldi20}, and is applied when an hourglass pattern is detected.

\paragraph{Preliminary notations}
Let us denote with $E$ a set of integral iteration vectors. The cardinality of such a set $|E|$ is the number of integer points inside this set.

Given a dimension $k$, $\phi_k : (i,j,k) \mapsto (k)$ is the projection to the dimension $k$. This notation can be extended to several dimensions. For example, $\phi_{j,k}: (i,j,k) \mapsto (j,k)$.

Given a set $E$, a dimension $k$ and a value $k_0$, $E_{k=k_0}$ is a slice of $E$, i.e., the set of points of $E$ whose value along the dimension $k$ is $k_0$: $E_{k=k_0} = \{ (i,j,k) \in E ~|~ k=k_0 \}$.
We extend this notation to several dimensions, as in $E_{k=k_0,j=j_0}$.
In the rest of the paper, we use a compact notation for slices, by omitting the dimension when it is not ambiguous, like $E_{k_0,j_0}$.

\paragraph{Starting point and intuition}
The goal of the proof is to find an upper bound on the size of a $K$-bounded set, to transform it into a lower bound on the data volume required by a portion of the program.
So, we start the proof with an arbitrary convex $K$-bounded set $E$ containing instances of the broadcast statement $SX[\vec{k}, \vec{j}, \vec{i}]$, which we assume is part of an hourglass pattern. Just like in the classical proof, an analysis of the dependencies of the program yields a set $\Phi$ of projections $\phi_{\vec{x}}$ for some dimensions $\vec{x}$, which project on the inset of $E$.

The classical proof involves applying the Brascamp-Lieb theorem on the $K$-bounded set $E$, by bounding the size $|\phi_{\vec{x}}(E)|$ of each of these projections by $K$.

In our proof, we split $E$ into two parts (Section~\ref{subsec:gen_iolb_proof_part1}) and we adapt the set of projections when we apply the Brascamp-Lieb theorem on each of these parts (Sections~\ref{subsec:gen_iolb_proof_part2} and~\ref{subsec:gen_iolb_proof_part3}) to obtain more precise bounds.

\paragraph{Running example} In the case of MGS, by following the chain of dependencies associated with the accesses $A[i][j]$, $Q[i][k]$ and $R[k][j]$ of statement $SU$, we infer that the projections are $\phi_{i,j}$, $\phi_{i,k}$ and $\phi_{k,j}$.
When following the classical proof, the application of the Brascamp-Lieb theorem results in the following inequality:
$$
|E| \leq |\phi_{i,j}(E)|^{\frac{1}{2}} . |\phi_{i,k}(E)|^{\frac{1}{2}} . |\phi_{k,j}(E)|^{\frac{1}{2}} \leq K^{3/2}
$$
We prove in Section~\ref{subsec:gen_iolb_proof_part4} a tighter upper bound: $|E| \leq \frac{K^2}{M} + 2K$.

\subsection{Part 1 - Decomposition of the K-bounded set}
\label{subsec:gen_iolb_proof_part1}

In the first part of the proof, we decompose $E$ into the union of two fragments: (i) $I'$ which has volume along the $\vec{k}$ dimensions, and (ii) $F$ which is flat along the $\vec{k}$ dimensions.
Both parts have their own upper bound, obtained with different reasoning, that is described in Section~\ref{subsec:gen_iolb_proof_part2} and Section~\ref{subsec:gen_iolb_proof_part3}.

\paragraph{Decomposition of $E$}
We consider the number of different values of $\vec{k}$ for a given value of $\vec{j}$:
$$
Tick_{\vec{j}} = \{\vec{k} \mid \exists \vec{i}, ~ SX[\vec{k},\vec{j},\vec{i}] \in E_{\vec{j}} \}.
$$
We split $E$ into two sets:
\begin{itemize}
    \item $E' = \cup_{\vec{j} \in J_{3+}} E_j$ where $J_{3+}=\{\vec{j} \mid 3\leq Tick_{\vec{j}}  \}$, and
    \item $E'' = \cup_{\vec{j} \in J_{12}} E_j$ where $J_{12}=\{\vec{j} \mid 1 \leq Tick_{\vec{j}} \leq 2 \}$.
\end{itemize}

The intuition between this separation is that (i) $E'$ contains the connected components which need to include an entire line of statement instances along the $\vec{i}$ dimensions, and (ii) $E''$ contains the connected components which are ``flat'' along the $\vec{k}$ dimensions.

We start by focusing on $E'$, to show that these components need to include a parametric number of iterations along the $\vec{i}$ dimension, using the hourglass pattern.

\begin{lemma}[Structure of $E'$]
    Given some $\vec{j} \in \phi_{\vec{j}}(E')$, let us consider the slice $E'_{\vec{j}}$ along the dimensions $\vec{j}$.
    Let us consider:
    \begin{itemize}
        \item $\vv{\kmin}(\vec{j}) = min_{\vec{k}} \{ \vec{k} \mid SX[\vec{k}, \vec{j}, \vec{i}] \in E'_{\vec{j}} \}$,
        \item $\vv{\kmax}(\vec{j}) = max_{\vec{k}} \{ \vec{k} \mid SX[\vec{k}, \vec{j}, \vec{i}] \in E'_{\vec{j}} \}$,
        \item $\vec{a}$, any subset of the indices of $\vec{i}$,
        \item $|\phi_{\vec{a}}(\mathcal{D}_S)| \geq W_{\vec{a}}$,
        a lower bound on the size of the projection of $\mathcal{D}_S$, the iteration domain of the statement $SX$.
    \end{itemize}

    Then:
    \begin{enumerate}
        \item $E'_{\vec{j}}$ is a connected component.
        \item For all $\vv{\kmin}(\vec{j}) <_{lex} \vec{k} <_{lex} \vv{\kmax}(\vec{j})$, $|\phi_{\vec{a}}(E'_{\vec{j}, \vec{k}})| \geq W_{\vec{a}}$.
    \end{enumerate}
    \label{lemma1_gen}
\end{lemma}

\begin{proof}
    (1) Because $SX$ satisfies the hourglass pattern properties, for any two instances $SX[\vec{k},\vec{j}, \vec{i}]$ and $SX[\vec{k'},\vec{j}, \vec{i'}]$ in $E'_{\vec{j}}$ where $\vec{k}<_{lex}\vec{k'}$, we can prove that there is a chain of dependencies from one statement instance to another.
    Therefore, using the convexity property of $E$, we conclude that $E'_j$ is a connected component.
    
    (2) By considering one instance of index $\vec{k} = \vv{\kmin}$ and another of index $\vec{k} = \vv{\kmax}$, we can show that all the $SX[\vec{i}, \vec{j}, \vec{k}]$ are in the middle of a dependency chain between the two of them.
    Therefore, by projecting the $\vec{i}$ on the subset of dimension $\vec{a}$, we conclude:
    $$
    \forall \vv{\kmin}(\vec{j}) <_{lex} \vec{k} <_{lex} \vv{\kmax}(\vec{j}),~ |\phi_{\vec{a}}(E'_{\vec{j}, \vec{k}})| \geq |\phi_{\vec{a}}(\mathcal{D}_S)| \geq W_{\vec{a}}.
    $$
\end{proof}

We consider $E' = I' \uplus B'$ defined by:
\begin{itemize}
    \item $I' = \cup_{\vec{j}\in J_{3+}} (E_{\vec{j}} - E_{\vec{j},\vv{\kmin}(\vec{j})} - E_{\vec{j},\vv{\kmax}(\vec{j})})$: the ``inside'' of the connected components of $E'$, according to dimensions $\vec{k}$.
    \item $B' = \cup_{\vec{j}\in J_{3+}} (E_{\vec{j},\vv{\kmin}(\vec{j})} \cup E_{\vec{j},\vv{\kmax}(\vec{j})})$: the ``boundaries'' of the connected components of $E'$, according to dimensions $\vec{k}$.
\end{itemize}

We define $F = B' \uplus E''$ the flat parts of $E$, and we adapt our previous decomposition of $E$ to obtain the desired decomposition:
$$
E = E' \uplus E'' = I' \uplus F.
$$

\subsection{Part 2 - Bound on the size of I'}
\label{subsec:gen_iolb_proof_part2}

Let us focus on the upper bound of the size of $I'$, using Lemma~\ref{lemma1_gen}.
The goal is to have tighter bounds on the projections involving some of the reduce/broadcast dimensions $\vec{i}$, to use them with the Brascamp-Lieb theorem, instead of the classical ``$\leq K$'' bound.


\begin{lemma}[Bounds on the size of some of the projections of $I'$]
Let us consider a projection $\phi_{\vec{x}, \vec{a}}\in \Phi$, where $\vec{a}$ is a subset of $\vec{i}$, and $\vec{x}$ is a subset of $(\vec{j}, \vec{k})$.
Assume that we have a lower bound $|\phi_{\vec{a}}(\mathcal{D}_S)| \geq W_{\vec{a}}$ on the size of the projection of the iteration domain $\mathcal{D}_S$ of statement $S$.
Then:
$$
|\phi_{\vec{x}}(I')| \leq \frac{K}{W_{\vec{a}}}.
$$
\label{lemma2_gen}
\end{lemma}

\begin{proof}
    Because $I'$ is a subset of a $K$-partition, then we have $|\phi_{\vec{a},\vec{x}}(I')| \leq K$.
    
    By slicing $I'$ along the dimensions of $\vec{x}$, we also have:
    $$
        |\phi_{\vec{a},\vec{x}}(I')| = |\phi_{\vec{a},\vec{x}}( \cup_{\vv{x''}} I'_{\vv{x''}} )| 
            = \sum\limits_{\vv{x''}} |\phi_{\vec{a},\vec{x}}(I'_{\vv{x''}})|.
    $$
    
    Furthermore:
    \begin{align*}
    |\phi_{\vec{a},\vec{x}}(I'_{\vv{x''}})| & \geq |\phi_{\vec{a}} I'_{\vv{x''}})| && \text{(a slice along $\vec{x}$ only has a single}\\
        &&& \text{point to be projected along $\vec{x})$}\\
        & \geq |\phi_{\vec{a}}(I'_{\vv{j''}, \vv{k''}})| \quad &&\text{($I'_{\vv{j''}, \vv{k''}} \subset I'_{\vv{x''}}$, for some $(\vv{j''},\vv{k''})$}\\
        &&& \text{matching $\vv{x''}$ on its dimensions)} \\
        & \geq W_{\vec{a}}. && \text{(by Lemma~\ref{lemma1_gen})}
    \end{align*}
        
    Therefore: $|\phi_{\vec{a},\vec{x}}(I')| \geq W_{\vec{a}} \times |\phi_{\vec{x}}(I')|$.
    In other words:
    $$
    |\phi_{\vec{x}}(I')| \leq \frac{|\phi_{\vec{a},\vec{x}}(I')|}{W_{\vec{a}}} \leq \frac{K}{W_{\vec{a}}}.
    $$
\end{proof}

To obtain a bound on $|I'|$, we use the set of projections $\Phi$, modified in the following way:
\begin{itemize}
    \item We add a projection on $\vec{i}$ whose bound is:
    $|\phi_{\vec{i}}(I')| \leq W$, where $W$ is the width of the hourglass.
    \item When a projection on $(\vec{x}, \vec{a})$ shares some of its dimensions $\vec{a}$ with $\vec{i}$, we use the projection on $\vec{x}$ instead, and the bound given by Lemma~\ref{lemma2_gen}.
    \item The rest of the projections, which do not involve dimensions of $\vec{i}$, are unchanged and associated with a classical upper bound $|\phi_{\vec{x}}(I')| \leq K$.
\end{itemize}
Then, we apply the Brascamp-Lieb theorem, while optimizing the values of the power $s$ of the size of the projections $|\phi(E)|$.
Notice that some projections can be filtered out through this optimization process ($s=0$), to obtain a tighter bound on $I'$.

\paragraph{Running example}

We notice that $\Phi$ contains two projections $\phi_{i,j}$ and $\phi_{i,k}$,
both of them using the dimension $i$ and another dimension.
Therefore, from Lemma~\ref{lemma2_gen}, we have the following bounds:
$$
 |\phi_j(I')| \leq \frac{K}{M} \quad\text{and}\quad |\phi_k(I')| \leq \frac{K}{M}.
$$

We apply the Brascamp-Lieb theorem on $I'$, using the projections on $i$, on $j$ instead of a projection on $(i,j)$, and on $k$ instead of a projection on $(i,k)$, each with coefficient $1$:
$$
|I'| \leq |\phi_i(I')| \times |\phi_j(I')| \times |\phi_k(I')|.
$$
By using our specialized bounds, we have:
$$
|I'| \leq M \times \frac{K}{M} \times \frac{K}{M} = \frac{K^2}{M}.
$$


\subsection{Part 3 - Bound on the size of F}
\label{subsec:gen_iolb_proof_part3}

We now focus on the bound of the remaining part of $E$, i.e., $F = (B' \uplus E'')$. 
Our goal is to exploit the fact that this set is ``flat'' along the $\vec{k}$ dimensions, to improve the upper bounds on the projections used in the Brascamp-Lieb application to $F$.
In particular, our proof will consider each $F_{\vec{j}}$ separately, the slice of $F$ for a given value $\vec{j}$. Then, by picking a list of well-chosen projections for the Brascamp-Lieb theorem, we obtain interesting bounds on the size of $F_{\vec{j}}$, that are finally summed to obtain our bound on $F$.

We recall that $E'$ and $E''$ contains the $\vec{j}$-slices of $E$ for the values of $\vec{j}$ such that $Tick_{\vec{j}}$ are respectively at least $3$ and at most $2$. Since $B' \subset E'$, $B'$ and $E''$ do not share the same set of $\vec{j}$.
So we get a ``flatness bound'': $\forall \vec{j} \in \phi_{\vec{j}}(F), ~ |\phi_{\vec{k}}(F_{\vec{j}})| \leq 2$.

In addition, we notice that because $F_{\vec{j}} \subset E$, for any projection $\phi \in \Phi$, $|\phi(F_{\vec{j}})| \leq |\phi(E)| \leq K$.


\medskip

In this part of the proof, instead of focusing on $F$, we apply the Brascamp-Lieb theorem to the slice $F_{\vec{j}}$.

As in Section~\ref{subsec:gen_iolb_proof_part2}, we start with the same list of projections $\Phi$ to the in-set of $E$, obtained by inspecting the chain of dependencies of the considered statement.
Then, we customize this list of projections to exploit the properties of $F_{\vec{j}}$:
\begin{itemize}
    \item We add a projection on $\vec{k}$ and the flatness bound $|\phi_{\vec{k}}(F_{\vec{j}})| \leq 2$.
    \item We identify a projection that involves a non-empty subset of the dimensions of $\vec{j}$: after applying the Brascamp-Lieb theorem, we will not try to immediately use an upper bound for this projection. Let us call $\phi_{\vec{w}}$ this projection.
    \item The remaining projections are left alone, and associated with the classical upper bound $|\phi_{\vec{x}}(F_{\vec{j}})| \leq |\phi_{\vec{x}}(F)| \leq K$.
\end{itemize}

At that point, we have obtained an upper bound on $|F_{\vec{j}}|$ of the following form:
$$
|F_{\vec{j}}| \leq e \times |\phi_{\vec{w}}( F_{\vec{j}} )| .
$$
where $e$ is a parametric expression using the parameter $K$ and independent of the value of $\vec{j}$.
Notice that the Brascamp-Lieb power above $|\phi_{\vec{w}}|$ must be equal to $1$, due to the fact that this is the only projection involving the $\vec{j}$ dimensions.


Let us consider the collection of projected sets $\phi_{\vec{w}}(F_{\vec{j}})$.
Two of these projected sets are either (i) identical, along the dimensions of $\vec{j}$ which are not present in $\vec{w}$, or (ii) disjoint.
Because the distinct values of these projections are always a subset of the inset of $E$, of size $K$, the sum of the size of the disjoint union of these distinct values is also bounded by $K$.
Let us call $R$ the number of values that can be taken by the dimensions of $\vec{j}$ which are not in $\vec{w}$. This is also the maximum number of times a $\phi_{\vec{w}}(E)$ projects on the same value.
Notice that $R=1$ when $\vec{w}$ covers all the dimensions of $\vec{j}$, and that $R$ can be a parametric expression in general.
%
So, we have:
$$
\sum_{\vec{j} \in \phi_{\vec{j}}(F)} |\phi_{\vec{w}}( F_{\vec{j}} )| \leq R \times K.
$$

Combining all these observations, we obtain the following upper bound on the size of $F$:
\begin{align*}
|F|  &= \sum_{\vec{j} \in \phi_{\vec{j}}(F)} |F_{\vec{j}}|   \leq   \sum_{\vec{j} \in \phi_{\vec{j}}(F)} e \times |\phi_{\vec{w}}( F_{\vec{j}} )|  \\
& \leq e \times  \sum_{\vec{j} \in \phi_{\vec{j}}(F)} |\phi_{\vec{w}}( F_{\vec{j}} )| \leq e \times R \times K .
\end{align*}

\paragraph{Running example}
In the MGS case, the flatness bound yields: $\forall j, |\phi_{k}(F_j)| \leq 2$.
We can apply the Brascamp-Lieb theorem on $F_j$, with the projection $\phi_{i,j}$, together with the projection $\phi_k$:
$$
|F_j| \leq |\phi_k(F_j)| \times |\phi_{i,j}(F_j)| \leq 2 \times |\phi_{i,j}(F_j)|.
$$

And since dimension $j$ is included in the dimensions $(i,j)$ of the projection, we have $R=1$ and therefore:
$$
|F| = \sum_{j \in \phi_j(F)} |F_j| \leq 2 \times \sum_{j \in \phi_j(F)} |\phi_{i,j}(F_j)| \leq 2K.
$$

\subsection{Part 4 - Wrapping things up}
\label{subsec:gen_iolb_proof_part4}

We have found in Section~\ref{subsec:gen_iolb_proof_part2} and Section~\ref{subsec:gen_iolb_proof_part3} an upper bound of the size of the two parts of $E$.
We simply sum them together to obtain an upper bound of the size of $E$, then apply Theorem~\ref{thm:STpartitioning} to deduce a lower bound on the data movement.

\paragraph{Running example}
Thanks to the above results, we can obtain lower bounds on the data movement of MGS.
\begin{theorem}[Lower bounds for MGS]
    The communication volume $Q$ for the MGS algorithm on a $M\times N$ matrix can be bounded as follows:
    $$\frac{M^2N(N-1)}{8(S+M)} \leq Q$$

    Furthermore, if $S\leq M$, we also have: 
    $$\frac{(M-S)N(N-1)}{4}\leq Q$$
    \label{thm:lb_mgs}
\end{theorem}
\begin{proof}
From the previous results, we have:
$$|E| = |I'| + |F| \leq \frac{K^2}{M} + 2K .$$

Then, by using Theorem~\ref{thm:STpartitioning} with $K=2S$:
$$
(K-S) \times \frac{MN(N-1)}{2\cdot\left(\frac{K^2}{M} + 2K\right)} = \frac{M^2N(N-1)}{8(S+M)}\leq Q
$$

Due to Lemma~\ref{lemma1_gen} and because we have at least an input dependency involving the dimension $i$, $|InSet(E')| > M$.
So, if we have $S \leq M$, then $E'$ must be empty. Therefore, $E = F$, and we can use only the second part of the bound: $|E|\leq 2K$.

Using Theorem~\ref{thm:STpartitioning} again, but this time with $K = M$, we obtain:
$$(K-S)\cdot\frac{MN(N-1)}{2 \times 2K} = (M-S)\cdot \frac{N(N-1)}{4}\leq Q$$
\end{proof}

\section{Experimental results - New lower bounds}
\label{sec:new_bounds}

In this section, we report the data movement lower bounds generated by IOLB for four kernels exhibiting an hourglass pattern. We compare the results using our technique (new bound) with those obtained without it (old bound).
These kernels are:
\begin{itemize}
    \item Modified Gram-Schmidt (Figure~\ref{fig:mgs_rl}), already used as a running example in Section~\ref{sec:gen_iolb_proof}.
    \item  QR Householder algorithm: both its A2V (Figure~\ref{fig:householder_a2v}) and V2Q parts (left-looking variants of respectively the GEQR2 and ORG2R subroutines in LAPACK~\cite{githubLAPACK}).
    \item Reduction to a bidiagonal matrix (GEBD2 subroutine)
    \item Reduction to a Hessenberg matrix (GEHD2 subroutine)
\end{itemize}


Figure~\ref{fig:IOLB_full_bounds} summarizes all the newly found full lower bounds, and Figure~\ref{fig:recap_bounds} focuses on their leading term, to emphasize their improvements.
More precisely:
\begin{itemize}
    \item Section~\ref{subsec:asymp_lb_mgs} contains an asymptotic analysis of the MGS lower bound.
    \item Section~\ref{subsec:qr_hh_lb} presents the lower bound for both QR Householder parts, and the GEBD2 computation.
    \item Section~\ref{subsec:gehd2_lb} presents the lower bounds for the GEHD2 computation.
\end{itemize}

Annex~\ref{sec:annex_ub_algo} contains a description of a tiled algorithm for MGS and for HH A2V and the computation of their amount of data movement.
This provides an upper bound to the minimal amount of data movement required by these algorithms, which matches asymptotically the provided lower bound.

\begin{figure}
    \centering
    \begin{tabular}{c|c|c}
        Kernel & Old bound~\cite{Olivry_pldi20} & New bound (hourglass)\\
        \hline
        MGS & $\Omega\left(\frac{MN^2}{\sqrt{S}}\right)$ & $\Omega\left(\frac{M^2N(N-1)}{S+M}\right)$\\
        QR HH A2V & $\Omega\left(\frac{MN^2}{\sqrt{S}}\right)$ & $\Omega\left(\frac{MN^2(N-M)}{N-M-S}\right)$ \\
        QR HH V2Q & $\Omega\left(\frac{MN^2}{\sqrt{S}}\right)$ & $\Omega\left(\frac{MN^2(N-M)}{N-M-S}\right)$ \\
        GEBD2 & $\Omega\left(\frac{MN^2}{\sqrt{S}}\right)$ & $\Omega\left(\frac{MN^2(M-N+1)}{8(S+M-N+1)}\right)$\\
        GEHD2 & $\Omega\left(\frac{N^3}{\sqrt{S}}\right)$ & $\Omega\left(\frac{N^4}{N+2S}\right)$\\
    \end{tabular}
    \caption{Summary of the new asymptotic data movement lower bounds.}
    \label{fig:recap_bounds}
\end{figure}

\begin{figure*}
    \centering
    
    \begin{tabular}{c|c}
        Kernel & Old bound~\cite{Olivry_pldi20}\\
        \hline
        MGS & $\frac{2M + 3MN + MN^2}{\sqrt{S}} + 5M - MN + \frac{7N-N^2}{2} - S - 6$\\
        QR HH A2V & $\frac{3MN^2 + 6M + 7N - N^3 - 9MN - 6}{3\sqrt{S}} + 5M -MN + 5N - S - 13$\\
        QR HH V2Q & $\frac{3MN^2 - N^3 + 6M + 7N - 9MN - 6}{3\sqrt{S}} + 2M + 2N + \frac{N - N^2}{2} -S - 4$\\
        GEBD2 & $\frac{3MN^2 - N^3 - 9MN + 6M + 7N - 6}{ 3\sqrt{S} } + 5N + 5M - MN - S - 13$\\
        GEHD2 & $\frac{5N^3 - 30N^2 + 55N - 30}{3\sqrt{S}} + \frac{69N-9N^2}{2} - 3*S - 56$\\
    \end{tabular}
    
    \begin{tabular}{c|c}
        Kernel & New bound (hourglass)\\
        \hline
        MGS & $\frac{N^2M^2 + 2M^2 - 3NM^2}{ 8(M+S) } + 5M - MN + \frac{7N-N^2}{2} - S - 6$ \\
        QR HH A2V  & $\frac{3MN^2 - 9MN + 7N + 6M - 6 - N^3}{24 (1 - \frac{S}{N-M})} + 5M - MN + 5N - S - 13$ \\
        QR HH V2Q & $\frac{3MN^2 - N^3 + 6M + 7N - 9MN - 6}{ 24 (1 + \frac{S}{M-N}) } + 2M + 2N + \frac{N - N^2}{2} - S - 4$ \\
        GEBD2 & $\frac{3MN^2 - N^3 + 3N^2 - 15MN + 4N + 18M - 12}{ 24 (1+ \frac{S}{1+M-N} ) } + 5N + 7M - MN - S - 18$\\
        GEHD2 & $\frac{N^3 - 6N^2 + 11N - 6}{ 12 (1 + \frac{S}{N-M-1} ) } - N^2 + 12N - S - 19$ \\
    \end{tabular}
    \caption{Data movement lower-bounds (with constants) automatically derived by IOLB~\cite{Olivry_pldi20} without/with hourglass detection.
    In GEHD2's new bound, a new parameter $M$ is introduced, corresponding to the place where we split the outer loop.
    Depending on $S$ and $N$, it can be instantiated with a different parametric expression (cf Section~\ref{subsec:gehd2_lb}).}
    \label{fig:IOLB_full_bounds}
\end{figure*}





%


\subsection{MGS - Asymptotic analysis}
\label{subsec:asymp_lb_mgs}

In this section, we analyze the bound obtained in Theorem~\ref{thm:lb_mgs} for MGS, by specializing it for different ordering of $M$ and $S$:
\begin{itemize}
    \item If $S\leq M/2$, we have $M/2\leq M-S$, so that the second bound yields:
    $$\frac{MN^2}{8} = \Omega(MN^2)\leq Q$$
    \item If $M/2\leq S$, we have $S+M\leq 3S$, so that the first bound becomes:
    $$\frac{M^2N^2}{24S} = \Omega\left(\frac{M^2N^2}{S}\right)\leq Q$$
\end{itemize}

The first result matches the amount of data movements obtained by the classical ordering of the MGS algorithm, as presented at the end of Section~\ref{subsec:hourglass_intuition}. Both the algorithm and the bound are thus asymptotically optimal when $S$ is small.

Demmel et al.~\cite{techreport-berkeley-demmel-grigori-hoemmen-langou-2008} propose a tiled ordering for the case $2M\leq S$ (in Section F.2 of their paper), that achieves an amount of data movement $O(M^2N^2/S)$, thus matching asymptotically the second upper bound. Both of these bounds are thus optimal up to a constant factor. For reference, the tiled ordering is detailed in Appendix~\ref{sec:appendix_mgs_algo}, together with the proof on its amount of data movement.

We can compare these bounds with the lower bound returned by the classical hourglass-less proof, whose asymptotic bound is $\Omega(\frac{MN^2}{\sqrt{S}})$. Our first bound for small values of $S$ is stronger by a factor of $Theta(\sqrt{S})$.
By writing our second bound as $\Omega(\frac{M}{\sqrt{S}}\cdot\frac{MN^2}{\sqrt{S}})$, we see that our bound is stronger by a factor of $\Theta(\frac{M}{\sqrt{S}})$. Since the input matrix has size $M \times N$, we can assume that $S< MN$, otherwise, the whole matrix fits in the cache and there is no need to minimize data transfers. Besides, we know that $N\leq M$, which leads to $S<M^2$.
We can conclude that $\frac{M}{\sqrt{S}} > 1$: our bound is asymptotically at least as strong as the previous bound.

The results of Theorem~\ref{thm:lb_mgs} can also be presented differently, to improve the constant in front of the dominant term. Indeed:
\begin{align*}
&\text{If $S\ll M$, the first bound yields} & \frac{MN^2}{4} &\leq Q,\\
&\text{Similarly, if $M\ll S$, the second result becomes}&  \frac{M^2N^2}{8S} &\leq Q.
\end{align*}

\subsection{Householder QR factorization and GEBD2}
\label{subsec:qr_hh_lb}

In this section, we present the bounds to both parts of the Householder QR factorization (LAPACK routines GEQR2 et ORG2R). The computation of the first part (A2V) is given in Figure~\ref{fig:householder_a2v} and the computation of the second part (V2Q) in Figure~\ref{fig:householder_v2q}.

\begin{figure}
\lstset{
  basicstyle=\footnotesize,
  xleftmargin=.07\textwidth, xrightmargin=.05\textwidth
}
\begin{lstlisting}[style=myC]
for ($k=N-1$; $k > -1$; $k\ -= 1$) {
    for ($j = k+1$; $j < N$; $j\ \inc 1$){
        $tau[j]$ = $0.0$;
        for($i = k+1$; $i < M$; $i\ \inc 1$){
SR:         $tau[j]$ += $A[i][k]$ * $A[i][j]$;
        }
    }
    for($j = k+1$; $j < N$; $j\ \inc 1$){ 
ST:     $tau[j]$ *= $tau[k]$;
    }
    $A[k][k]$ = $1.0$ - $tau[k]$;
    for($j = k+1$; $j < N$; $j\ \inc 1$){
        $A[k][j]$ = $-tau[j]$;
    }
    for($j = k+1$; $j < N$; $j\ \inc 1$){
        for($i = k+1$; $i < M$; $i\ \inc 1$){
SU:         $A[i][j]$ -= $A[i][k]$ * $tau[j]$;
        }
    }
    for($i = k+1$; $i < M$; $i\ \inc 1$){
        $A[i][k]$ = $- A[i][k]$ * $tau[k]$;
    }
}
\end{lstlisting}
\caption{QR Householder computation - Part V2Q (LAPACK subroutine ORG2R). We assume that $M \geq N$.}
\label{fig:householder_v2q}
\end{figure}


By using the hourglass reasoning, we obtain the following new lower bounds for both kernels.

\begin{theorem}[Lower bounds for HH - part A2V]
    The communication volume $Q$ for the A2V part of the HH algorithm on a $M \times N$ matrix, with $M>N$, can be bounded as follows:
    $$
    \frac{(3M-N)N^2(M-N)^2}{24(MS+(M-N)^2)} \leq Q
    $$
    If $M \gg N$, then the bound becomes:
    $$
    \frac{M^2 N(N-1)}{8(S+M)} \leq Q
    $$
\label{thm:lb_hh_a2v}
\end{theorem}

\begin{theorem}[Lower bound for HH - Part V2Q]
    The communication volume $Q$ for the HH V2Q algorithm on a $M \times N$ algorithm, with $M > N$, can be bounded as follows:
    $$
    \frac{N (N-1)(3M-N-1)(M-N)^2}{24 ( (M-N)^2 + SM )} \leq Q
    $$
    When $M \gg N$, this bound becomes:
    $$
    \frac{N(N-1)M^2}{8 (S+M)}\leq Q
    $$
\label{thm:lb_hh_v2q}
\end{theorem}


One interesting detail of the proof concerns the detection of the hourglass, and the criteria on its size (third criterion introduced in Section~\ref{subsec:hourglass_def}).
For the MGS computation, the size of its hourglass was constant and equal to $M$.
In the case of both Householder computations, the size of their hourglass is parametrized by the outer loop iteration value $k$, and is equal to $(M-1-k)$.
Its smallest value happens for $k=N-1$, and is $(M-N)$.
By using this value in Lemma~\ref{lemma1_gen}, we can derive the announced lower bound.

The lower bound proof of the GEBD2 subroutine is similar to both Householder proofs.

\begin{theorem}[Lower bounds for GEBD2]
    The communication volume $Q$ for the GEBD2 subroutine on a $M \times N$ matrix, with $M\geq N$, can be bounded as follows:
    $$
    \frac{MN^2(M-N+1)}{8(S+M-N+1)} \leq Q
    $$
    
    If $M \gg N$, then the bound becomes:
    $$
    \frac{M^2 N^2}{8(S+M)} \leq Q
    $$
\label{thm:lb_hh_a2v}
\end{theorem}

\subsection{Hessenberg matrix reduction}
\label{subsec:gehd2_lb}

\begin{figure}
\lstset{
  basicstyle=\footnotesize,
  xleftmargin=.03\textwidth, xrightmargin=.05\textwidth
}
\begin{lstlisting}[style=myC]
for ($j = 0$; $j<n-2$; $j \inc 1$) {
  $norma2$ = $0.0$;
  for ($i=j+2$; $i<n$; $i \inc 1$) {
    $norma2$ += $A[i][j]$ * $A[i][j]$;
  }
  $norma$ = $sqrt$( $A[j+1][j]$ * $A[j+1][j]$ + $norma2$);
  $A[j+1][j]$ = ($A[j+1][j] > 0$)?
        ($A[j+1][j]$+$norma$):($A[j+1][j]$-$norma$);
  $tau$ = $2.0$/($1.0$ + $norma2$/($A[j+1][j]$ * $A[j+1][j]$));
  for ($i = j+2$; $i < n$; $i \inc 1$ ) {
    $A[i][j]$ /= $A[j+1][j]$;
  }
  $A[j+1][j]$ = ($A[j+1][j] > 0$)?($-norma$):($norma$);
  for ($i = j+1$; $i < n$; $i \inc 1$) {
    $tmp[i]$ = $A[j+1][i]$;
    for ($k = j+2$; $k < n$; $k \inc 1$) {
      $tmp[i]$ += $A[k][j]$ * $A[k][i]$;
    }
  }
  for ($i = j+1$; $i < n$; $i \inc 1$) {
    $tmp[i]$ *= $tau$;
  }
  for ($i = j+1$ ; $i < n$ ; $i \inc 1$) {
    $A[j+1][i]$ -= $tmp[i]$;
  }
  for ($i = j+2$; $i < n$; $i \inc 1$) {
    for ($k = j+1$; $k < n$; $k \inc 1$) {
      $A[i][k]$ -= $A[i][j]$ * $tmp[k]$;
    }
  }
  for ($i = 0$; $i < n$; $i \inc 1$) {
    $tmp[i]$ = $A[i][j+1]$;
    for ($k = j+2$; $k < n$; $k \inc 1$) {
      $tmp[i]$ += $A[i][k]$ * $A[k][j]$;
    }
  }
  for ($i = 0$; $i < n$; $i \inc 1$) {
    $tmp[i]$ *= $tau$;
  }
  for ($i = 0$; $i < n$; $i \inc 1$) {
    $A[i][j+1]$ -= $tmp[i]$;
  }
  for ($i = 0$; $i < n$; $i \inc 1$) {
    for ($k = j+2$; $k < n$; $k \inc 1$) {
      $A[i][k]$ -= $tmp[i]$ * $A[k][j]$;
    }
  }
}
\end{lstlisting}
\caption{Hessenberg matrix factorization of a $N \times N$ matrix (LAPACK subroutine GEHD2)}
\label{fig:gehd2_code}
\end{figure}

In this section, we present the bound for the Hessenberg matrix factorization kernel (GEHD2 kernel in LAPACK), whose code is in Figure~\ref{fig:gehd2_code}.
%
By using the hourglass reasoning, we obtain the following new lower bound for the GEHD2 kernel.

\begin{theorem}[Lower bound for GEHD2]
    The communication volume $Q$ for the Hessenberg matrix factorization algorithm on a $N \times N$ algorithm can be bounded as follows:
    $$
    \frac{1}{12} . \frac{N^4}{N+2S} \leq Q
    $$
    
    When $N \gg S$,
    $$
    \frac{N^3}{24} \leq Q
    $$
\label{thm:lb_gehd2}
\end{theorem}

As in Section~\ref{subsec:qr_hh_lb}, the size of the hourglass depends on the value of iteration $j$ of the outermost loop (temporal dimension), and is equal to $(N-2-j)$, where $0 \leq j < N-2$.
This means that its minimal value is $1$, which causes an issue with our proof.

A way to solve this problem is to introduce a new parameter $M < N-2$ and to consider a loop splitting of the outermost temporal dimension $j$ at iteration $M$.
First, notice that a loop splitting does not change the dependencies of a program, thus the lower bound of the split version will apply to the original one.

The first half of the split program satisfies the hourglass pattern and has a minimum size of its hourglass of $(N-M-1)$. Thus, as long as this quantity is not bounded by a constant, the hourglass reasoning should apply.
The second half does not satisfy the last condition of the hourglass pattern, so the classical derivation is used there.
Notice that its bound is asymptotically worse than the bound on the first part, so the bound of the first part will dominate.
We consider $M=\frac{N}{2}-1$ to obtain the first bound, and $M=N-S-2$ to obtain the second bound when $N \gg S$.

\section{Related work}
\label{sec:rel_work}

\paragraph{I/O complexity and automation of the proof}
The seminal work of Hong and Kung~\cite{hong-pebble} introduced the red-blue pebble game and used it as a formalism to manually prove the I/O lower bound of programs.
Following this contribution, many papers~\cite{aggarwal-sort,demmel-linearalgebra,demmel-matmult,Irony04,kwasniewski-matmult,doi:10.1137/080731992,julien-matmult,savage-fft,savage-stencil,Scott15,Bilardi18} focused on the manual proof of lower bounds for various (classes of) programs.
A large portion of these papers consider a formalism that forbid recomputation. This assumption is necessary to decompose complex CDAGs into simpler subregions, and to be able to recombine the bounds of each region.

Irony et al.~\cite{Irony04} used the Loomis-Whitney bound, a specialization of the Brascamp-Lieb theorem with canonical projections, on the \texttt{gemm} algorithm.
This was later extended by the work of Christ et al.~\cite{Christ13} for the class of affine programs, by using the Brascamp-Lieb theorem.
Then, Elango et al.~\cite{Elango15} added the idea of considering paths of dependencies, completing the $K$-partitioning method as introduced in Section~\ref{sec:background} of our paper.

This leads to the work of Olivry et al.~\cite{Olivry_pldi20} that introduced the first automatic lower bound derivation tool, based on combining bounds obtained from the $K$-partitioning and the wavefront methods of proof.
This tool includes several refinements to the bound derivation, by taking advantage of some specific properties.
For example, if the projections use disjoint parts of the inset, the derived bound could be improved by a constant factor.
A later extension~\cite{Olivry_pldi21} exploits the fact that some dimensions have very small sizes and applies this reasoning to convolutions.
%

\section{Conclusion}
\label{sec:concl}

In this paper, we introduced a new proof reasoning to derive data movement lower bounds, based on a pattern of dependencies, called the hourglass pattern.
This pattern appears on several linear algebra kernels, such as the Modified Gram-Schmidt, the QR Householder decomposition (A2V and V2Q parts), the bidiagonal matrix reduction, and the Hessenberg matrix reduction.
We managed to derive asymptotically tighter lower bounds for these kernels, compared to the previous classical methods.
We also provided tiled algorithms for MGS and QR A2V, whose amount of data movement matches asymptotically the new lower bounds, up to a constant factor.

We integrated the hourglass detection and its associated proof in the automatic data movement lower bound derivation tool of Olivry et al., IOLB~\cite{Olivry_pldi20}, whose implementation is freely available.



\paragraph{Individual contributions}
Fabrice identified the original problem on MGS and the intuition on how to attack this problem.
The implementation in IOLB and the proof were written by Guillaume. The proof was refined by Lionel, who also provided a detailed study of the bound, and an experimental performance analysis.
Julien identified the interesting kernels exhibiting an hourglass pattern and provided valuable insights on the linear algebra side. Lionel and Julien are also at the source of the upper bound proof in the appendix.

\paragraph{Acknowledgments} Julien was partially supported for this work by NSF award \#2004850.


\bibliographystyle{abbrv}
\bibliography{biblio,biblio_nla}

\newpage  

\appendix


\section{Data movement upper bound - tiled algorithms}
\label{sec:annex_ub_algo}

In this Appendix, we present the tiled algorithm for the MGS and Householder A2V kernels.
We compute the data movement of these algorithms to obtain an upper bound on the minimal amount of data movement needed, and we show how it relates to their corresponding lower bound.

\subsection{Tiled algorithm for MGS}
\label{sec:appendix_mgs_algo}

As shown in Section~\ref{sec:gen_iolb_proof}, we have an improved asymptotic lower bounds for MGS when $M \ll S$.
For reference and completeness, we show in Figure~\ref{fig:opt_code_mgs__ll} a tiled left-looking ordering, based on the ideas from~\cite[\S9.1.4]{techreport-berkeley-demmel-grigori-hoemmen-langou-2008}. Then, we provide the proof that this ordering asymptotically matches this bound.

Let us show that, if we choose the block size $B$ such that $(M+1)B < S$, then this algorithm satisfies the following properties:
\begin{enumerate}
    \item it does not spill out of memory
    \item the amount of ``read'' operations is $\frac{1}{2} \frac{ M^2  N^2 }{ S }$ (leading term).
    \item the amount of ``write'' is $MN + \frac{1}{2}N^2$ (leading terms).
    \item the total amount of I/O (read+write) is $\frac{1}{2} \frac{ M^2  N^2 }{ S }$ (leading term).
\end{enumerate}

The input of this algorithm is an $M$-by-$N$ matrix $A$. The outputs are the $M$-by-$N$ matrix $Q$ (stored in the array $A$) and the upper triangular matrix $R$.
The first for-loop (line 1) moves along the columns of $A$ by block. At step $j0$, we work with the block $A(0{:}M,j0{:}j0+B)$.
We consider the left-looking variant of the algorithm. So, at the start of the $j0$ iteration, the block $A(0{:}M,j0{:}j0+B)$ contains the initial data of the matrix $A$, and at the end of the $j0$ iteration, the same block contains the final data of the matrix $Q$.
The block $A(0{:}M,j0{:}j0+B)$ stays in memory during the $j0$ iteration, so we read it at the start of the $j0$ iteration and write it at the end of this iteration.

\begin{figure}
\lstset{
  basicstyle=\footnotesize,
  xleftmargin=.03\textwidth, xrightmargin=.05\textwidth
}
\begin{lstlisting}[style=myC]
for ($j0 = 0$; $j0 < N$; $j0\ \inc B$) {
// read A(1:M,j0:j0+B)
   for ($i = 0$; $i < j0$; $i\ \inc 1$) {
//    read A(1:M,i)
      for ($j = j0$; $((j < j0+B)\&\&(j < N))$; $j\ \inc 1$) {
         $R[i][j]$ = $0.0$; 
         for ($k = 0$; $k < M$; $k\ \inc 1$)
            $R[i][j]$ += $A[k][i]$ * $A[k][j]$;
         for ($k = 0$; $k < M$; $k\ \inc 1$)
            $A[k][j]$ -= $A[k][i]$ * $R[i][j]$;
      }
//    discard A(1:M,i)
   }
   for ($j = j0$; $((j < j0+B)\&\&(j < N))$; $j\ \inc 1$) {
      for ($i = j0$; $i < j$; $i\ \inc 1$) {
         $R[i][j]$ = $0.0$;
         for ($k = 0$; $k < M$; $k\ \inc 1$)
            $R[i][j]$ += $A[k][i]$ * $A[k][j]$;
         for ($k = 0$; $k < M$; $k\ \inc 1$)
            $A[k][j]$ -= $A[k][i]$ * $R[i][j]$;
      }
      $R[j][j]$ = $0.0$;
      for ($k = 0$; $k < M$; $k\ \inc 1$)
         $R[j][j]$ += $A[k][j]$ * $A[k][j]$;
      $R[j][j]$ = sqrt($R[j][j]$);
      for ($k = 0$; $k < M$; $k\ \inc 1$)
         $A[k][j]$ /= $R[j][j]$;
   }
// write A(1:M,j0:j0+B)
}
\end{lstlisting}
\caption{Modified Gram-Schmidt - left-looking tiled code with $\mathcal{O}\left(\frac{M^2N^2}{S}\right)$ data movements, if the block size $B$ satisfies $(M+1)B< S$.}
\label{fig:opt_code_mgs__ll}
\end{figure}

For each vector column on the left of this block of columns (the $A(0{:}M,i)$ where $i<j0$), we first need to project our block of $A$, $A(0{:}M,j0{:}j0+B)$, onto the orthogonal complement of this column.
This is done by the first loop nest (line 3), which performs this projection one column at a time.
This requires the program to load the column $A(0{:}M,i)$.
Since we are using the left-looking version of the algorithm, we only need to ``read'' $A(0{:}M,i)$ and we do not need to ``write'' $A(0{:}M,i)$ back in memory.

Now, let us count the number of I/O accesses of this program.
To fit the $(B+1)$ columns of $A$ in cache (block $A(0{:}M,j0{:}j0+B)$ and column $A(0{:}M,i)$), we assume that $B$ is chosen such that $MB + M < S$.
Over the whole course of the algorithm, moving the blocks $A(1{:}M,j0{:}j0+B)$ into cache only happens once per block, so the I/O cost for these movements is only $MN$ read and $MN$ write operations.
The main I/O cost happens when reading the columns $A(0{:}M,i)$ one at a time, inside the very first two for-loops (lines 1 and 3, on iterations $j0$ and $i$).
    
Notice that the I/O benefit of the blocked algorithm occurs here: every time we move the column $A(1{:}M,i)$ in memory, we reuse it to project over $B$ columns at once, instead of projecting over a single column in the unblocked version.
This leads to a reduction of the I/O by a factor $B$.
More precisely, the total amount of data movement involved in the reading of column $A(0{:}M,i)$ is:
$$
\begin{array}{l}
    M \times \left|\{ j0,i \mid 0\leq j0 <N \text{ and } j0 \text{ mod } B = 0 \text{ and } 0\leq i< j0 \}\right|\\
    \quad = M \times \left|\{ j0,i \mid 0\leq j0 <N/B \text{ and } 0\leq i< j0.B \}\right|\\
    \quad = M \times \sum\limits_{j0=0}^{N/B} j0 \times B ~\approx~ MB\left(\frac{N^2}{2B^2}\right) ~=~ \frac{1}{2} \frac{MN^2}{B}.
\end{array}
$$
The total data movement for the algorithm reading the columns $A(1:M,i)$, 
$\textmd{for }(i = 0; i < j0; i++)$, $\textmd{for }(j0 = 0; j0 < N; j0+=B)$ is:
$$
\begin{array}{l}
    \sum_{(j0 = 0; j0 < N; j0+=B)}  \sum_{(i = 0; i < j0; i++)} M \\
    \hspace{1cm} = M \sum_{(j0 = 0; j0 < N; j0+=B)} j0 \\
    \hspace{1cm} \approx M B \sum_{(j0 = 0; j0 < \frac{N}{B}; j0++)} j0 \\
    \hspace{1cm} \approx \frac{1}{2} M B \left(\frac{N}{B}\right)^2 = \frac{1}{2} \frac{MN^2}{B}.
\end{array}
$$

This is indeed a factor of $B$ lower than the number of reads of the unblocked algorithm.
    
Because we have chosen $B$ such that $M ( B  + 1)< S$, we can take its maximal value to minimize the number of reads: $B = \lfloor \frac{S}{M} \rfloor - 1 \approx \frac{S}{M}$.
Thus, the I/O cost of this tiled algorithm is, assuming $B \geq 1$:
$$I/O \approx \frac{1}{2} \frac{M^2N^2}{S}.$$

\paragraph{Additional remarks}
This algorithm works for any block size $B$, any $M$, any $N$ ($N < M$), and any $S$. When we make the analysis of data movement, we further assume that $M (B  + 1)< S$.
So, for the I/O analysis to be correct, we need at least $S > 2M$. In other words, we assume that at least two columns fit in cache.

Also, we do not consider $R$ in the data movement or in the cache constraint.
Once a value of $R$ is computed, this value is used only once, and right after its computation.
We can thus write these values as soon as we are done with them. Hence, (1) the total number of writes due to $R$ is $N (N+1) / 2 \approx N^2/2$.  This $N^2/2$ ``write'' I/O is negligible for the algorithm (and cannot be optimized anyway), so we will  neglect it; and (2) the requirement of cache for $R$ is $1$ for the whole computation, so we do not consider $R$ at all for the cache usage.

Note that it is also possible to follow the same reasoning to write a tiled version of the right looking variant of the algorithm.
We get a volume of I/O with a similar order of magnitude, but the constant is higher.
The right-looking variant performs more I/O than the left-looking variant, but, in addition, writes are more expensive than reads, and the I/Os of the right-looking variant are dominated by writes, whereas, the I/Os of the left-looking variant are dominated by reads. 
So, in the sequential context, we prefer to use the left-looking variant. There are some  advantages in using the right-looking variants in the sense that the updates can be done in parallel, so in a parallel multicore setting, for example, the right-looking variants would be interesting.

\subsection{Tiled algorithm for Householder A2V}
\label{sec:appendix_hh_a2v_algo}

\begin{figure}
\lstset{
  basicstyle=\footnotesize,
  xleftmargin=.05\textwidth, xrightmargin=.05\textwidth
}
\begin{lstlisting}[style=myC]
for ($k0 = 0$; $k0 < N$; $k0\ \inc B$) {
// read A(1:M,k0:k0+B)
  for($j = 0$; $j < k0$; $j\ \inc 1$){
//    read A(j:M,j)
    for ($k = k0$; $(k < k0+B)\&\&(k < N)$; $k\ \inc 1$) {
      $tmp$ = $A[j][k]$;
      for($i = j+1$; $i < M$; $i\ \inc 1$)
        $tmp$ += $A[i][j]$ * $A[i][k]$;
      $tmp$ = $tau[j]$ * $tmp$;
      $A[j][k]$ = $A[j][k]$ - $tmp$;
      for($i = j+1$; $i < M$; $i\ \inc 1$)
        $A[i][k]$ = $A[i][k]$ - $A[i][j]$ * $tmp$;
    }
//    discard A(j:M,j)
  }
  for ($k = k0$; $((k < k0+B)\&\&(k < N))$; $k\ \inc 1$) {
    for($j = k0$; $j < k$; $j\ \inc 1$) {
      $tmp$ = $A[j][k]$;
      for($i = j+1$; $i < M$; $i\ \inc 1$)
        $tmp$ += $A[i][j]$ * $A[i][k]$;
      $tmp$ = $tau[j]$ * $tmp$;
      $A[j][k]$ = $A[j][k]$ - $tmp$;
      for($i = j+1$; $i < M$; $i\ \inc 1$)
        $A[i][k]$ = $A[i][k]$ - $A[i][j]$ * $tmp$;
    }
    $norma2$ = $0.0$;
    for($i = k+1$; $i < M$; $i\ \inc 1$)
      $norma2$ += $A[i][k]$ * $A[i][k]$;
    $norma$ = $sqrt$( $A[k][k]$ * $A[k][k]$ + $norma2$ );
    $A[k][k]$ = ($A[k][k] > 0.0$)?($A[k][k]$+$norma$)
            :($A[k][k]$-$norma$);
    $tau[k]$ = $2.0$/
        ($1.0$ + $norma2$/($A[k][k]$*$A[k][k]$));
    for($i = k+1$; $i < M$; $i\ \inc 1$)
      $A[i][k]$ /= $A[k][k]$;
    $A[k][k]$ = ($A[k][k]$ > $0.0$)?(-$norma$):($norma$);
  }
// write A(1:M,k0:k0+B)
}
\end{lstlisting}
\caption{QR Householder A2V part - Tiled implementation, with $\mathcal{O}\left(\frac{M^2N^2}{S}\right)$ data movements,
assuming $2M < S$.}
\label{fig:opt_code_hh__ll}
\end{figure}

Figure~\ref{fig:opt_code_hh__ll} shows an implementation that matches asymptotically this bound. It is based on the same idea as the tiled algorithm for MGS in Section~\ref{sec:appendix_mgs_algo}, but adapted to the A2V algorithm.

In particular, we show that, if we chose $B$ such that $(M+1)B < S$, then our algorithm satisfies the following properties:
\begin{enumerate}
    \item it does not spill out of memory.
    \item the amount of ``read'' is $\frac12 \frac{M^2N^2 -MN^3/3}{S}$.
    \item the amount of ``write'' is $MN$.
    \item the total amount of I/O (read+write) is $\frac12 \frac{M^2N^2 - MN^3/3}{S}$.
\end{enumerate}
We have only show the leading terms for these previous I/O costs.

The input of our algorithm is the $M$-by-$N$ matrix $A$. The outputs are the $N$-by-$N$ upper triangular matrix $R$ and the  $M$-by-$N$ unit lower triangular matrix $V$, that are computed in place in the $M$-by-$N$ array $A$.
The first for-loop (line 1), using iteration $k0$, moves along the column of $A$ by block.
At step $k0$, we work on the block $A(0{:}M,k0:k0+B)$.
We consider the left-looking variant of the algorithm.
So, at the start of the $k0$ iteration, the block $A(0{:}M,k0:k0+B)$ contains the initial data of the matrix $A$; and at the end of the $k0$ iteration, 
it contains the final data ($V$ and $R$).
This block remains in cache during the $k0$ iteration, so we read it at the start of the $k0$ iteration, and we write it at the end of this iteration.

For each vector column on the left of this block (the $A(0{:}M,j)$ where $j<k0$), we need to reflect our current block $A(j{:}M,k0{:}k0+B)$ with the elementary Householder reflectors defined by $A(j+1{:}M,j)$
and $tau[j]$, one column at a time.
This requires the algorithm to load $A(j+1{:}M,j)$.
Since we are using the left-looking version of the algorithm, we do not need to ``write'' it back in memory.

Now, let us count the number of I/O accesses of this program.
To fit the $(B+1)$ columns of $A$ in cache, we assume that $B$ is chosen such that $M(B + 1) < S$.
Over the whole course of the algorithm, moving the blocks $A(1{:}M,k0{:}k0+B)$ in cache happens only once per block, so the I/O for this is $2MN$ and is negligible.
The main I/O cost happens when reading the columns $A(j+1{:}M,j)$, one at a time, inside the two very first for-loops (using iteration $k0$ and $j$).

Notice that the I/O gain of this blocked algorithm against the unblocked version occurs there:
every time we load a column $A(j+1{:}M,j)$, we reuse it to reflect over $B$ columns at once, instead of a single column for the unblocked version.
This leads to a reduction of the I/O cost by a factor of $B$.
More precisely, the total reads of columns $A(j+1{:}M,j)$ in the algorithm is:
$$
\begin{array}{l}
\sum_{k0 = 0,~ k0+=B}^{N-1} \sum_{j = 0}^{k0-1} (M - j) ~\approx~ \sum_{k0 = 0}^{N/B} \sum_{j = 0}^{B.k0} (M - j)\\
\quad \quad \approx~ \sum_{k0 = 0}^{N/B} \left( B.k0.M - \frac{(B.k0)^2}{2} \right) \\
\quad \quad =~ B.M. \left(\sum_{k0 = 0}^{N/B} k0\right) - \frac{B^2}{2} . \left(\sum_{k0 = 0}^{N/B} k0^2\right)\\
\quad \quad \approx~ B.M.\frac{N^2}{2.B^2} - \frac{B^2}{2} . \frac{N^3}{3B^3} ~=~ \frac{MN^2}{2B} - \frac{N^3}{6B}\\
\quad \quad \approx~ \frac12 \frac{MN^2 - N^3/3}{B}.
\end{array}
$$
This is indeed a factor of $B$ less than the unblocked algorithm.

Because we require $B$ to satisfy $ M ( B  + 1)< S$, we choose $B = \lfloor \frac{S}{M} \rfloor - 1 \approx \frac{S}{M}$.
So the I/O cost of our algorithm is, assuming $B \geq 1$,
$$I/O \approx \frac12 \frac{M^2N^2 - N^3/3}{S}.$$

\paragraph{Additional remark}
Just like for the MGS algorithm, the I/O analysis requires $S>2M$, which means that we assume that at least two columns fit in cache.

\end{document}